\documentclass[a4paper,12pt]{article}

\usepackage[utf8]{inputenc} 
\usepackage[T1]{fontenc}
\usepackage{url}
\usepackage{ifthen}
\usepackage{multicol}
\usepackage{amsfonts, amsmath, amssymb, amsthm}
\usepackage{array, makecell}
\usepackage{pgf, tikz, color, xcolor, soul}
\usepackage{bbold}
\usepackage{bm}
\usepackage{enumitem}
\usepackage{caption}
\usepackage{subcaption}
\usepackage{cite}
\usepackage[title]{appendix}

\usepackage{amsfonts, amsthm}
\usepackage{ dsfont }
\newcommand{\x}{\boldsymbol{x}}
\newcommand{\w}{\boldsymbol{w}}
\newcommand{\f}{\boldsymbol{f}}
\newcommand{\y}{\boldsymbol{y}}
\newcommand{\e}{\boldsymbol{e}}
\renewcommand{\u}{\boldsymbol{u}}
\newcommand{\h}{\boldsymbol{h}}
\renewcommand{\S}{\boldsymbol{S}}
\newcommand{\bLambda}{\boldsymbol{ \Lambda}}
\newcommand{\bDelta}{\boldsymbol{ \Delta}}
\newcommand{\supp}{supp}

\newcommand{\eps}{\boldsymbol{\varepsilon}}
\newcommand{\Bin}{Bin}

\providecommand{\abs}[1]{\lvert#1\rvert} 
\providecommand{\norm}[1]{\lVert#1\rVert}
\usepackage[normalem]{ulem}
\newcommand\redsout{\bgroup\markoverwith{\textcolor{red}{\rule[0.5ex]{2pt}{0.4pt}}}\ULon}

\newtheorem{theorem}{Theorem}
\newtheorem{proposition}{Proposition}
\newtheorem{corollary}{Corollary}
\newtheorem{lemma}[theorem]{Lemma}

\newtheorem{remark}{Remark}%

\newtheorem{definition}{Definition}%

\title{Complete Traceability Multimedia Fingerprinting Codes Resistant to Averaging Attack and Adversarial Noise with Optimal Rate}
\author{Ilya Vorobyev \\
	Skolkovo Institute of Science and Technology
}

\date{}

\begin{document}

	\maketitle

	
	\abstract{In this paper we consider complete traceability multimedia fingerprinting codes resistant to averaging attacks and adversarial noise. Recently it was shown that there are no such codes for the case of an arbitrary linear attack. However, for the case of averaging attacks complete traceability multimedia fingerprinting codes of exponential cardinality resistant to constant adversarial noise were constructed in~\cite{egorova2020existence}. We continue this work and provide an improved lower bound on the rate of these codes.}
	

	
	
	\maketitle
	
	\section{Introduction}\label{sec::intro}
	
	Multimedia fingerprinting codes are used to protect digital content from illegal copying and redistribution. 
	The key idea of this technique is to embed a unique signal, called watermark, into every copy, so that it can be tracked to its buyer~\cite{liu2005multimedia,trappe2003anti}. Watermarks should be able to protect the dealer from collusion attack, when a coalition of dishonest users (pirates) construct a new file, for example, by averaging their copies of the same content. By gathering a big enough coalition it is possible to sufficiently decrease the impact of each individual fingerprint, which makes it hard for the dealer to identify the pirates. In papers~\cite{cheng2011anti,cheng2011separable} the authors propose to use separable (or signature) codes to track all members of the coalition. 
	
	A model of multimedia fingerprinting with an adversarial noise was proposed in~\cite{egorova2019signature}, i.e. the coalition of dishonest users can add some noise to the content in order to hide their fingerprints. In~\cite{fan2020signature} it was shown that there are no multimedia codes resistant to a general linear attack and an adversarial noise. However, in~\cite{egorova2020existence} the authors proved that for the most common case of averaging attack one can construct multimedia codes with a non-vanishing rate.  We continue their research and prove a new lower bound on the rate, which has the same order as an upper bound. A detailed survey of state-of-the-art results can be found in~\cite{egorova2021separable}.
	
	The rest of the paper is structured as follows. In Section~\ref{sec::problem statement} we introduce the required notation and definitions and formally describe the problem. Our main result is proved
	in Section~\ref{sec::main result}. Section~\ref{sec::conclusion} concludes the paper and discusses some open problems.
	
	\section{Problem statement}\label{sec::problem statement}
	Vectors are denoted by bold letters, such as $\x$, 
	and the $i$th entry is referred to as $x_i$. The set of integers $\{1,\,2,\,\ldots,\,M\}$ is abbreviated by $[M]$. The sign $\norm{\cdot}$ stands for the Euclidean norm. A support $\supp(\x)$ of a vector $\x$ is a set of such coordinates $i$ that $x_i\ne 0$. Scalar (dot) product of vectors $\x$ and $\y$ is denoted as $\langle \x,\,\y\rangle$, greatest common divisor of integers $a$ and $b$ is referred to as $(a, b)$. For a given binary $n\times M$ matrix $H$ with columns $\h_1,\,\ldots,\,\h_M$ and set $I\subset [M]$ introduce notation for a result of averaging attack 
	$$\sigma (H\mid I)=\abs{I}^{-1}\sum\limits_{i\in I}\h_i.$$
	A binary entropy function $h(x)$ is defined as follows
	$$
	h(x)=-x\log_2x - (1-x)\log_2(1-x).
	$$ 
	Suppose that multimedia content is represented by a vector $\x\in \mathbb{R}^N$, which is being sold to $M$ users. Vector $\x$ is often called a host signal. To protect the content from unauthorized copying the dealer constructs a set of watermarks $\w_1, \,\ldots, \,\w_M$, which are also called fingerprints. The dealer fixes $n$ orthonormal vectors $\f_1, \,\ldots, \,\f_n$ of length $N,\,\f_i\in \mathbb{R}^N$ and forms watermarks $\w_i$ as linear combinations of $\f_j$ with binary coefficients $h_{ij}\in\{0,\,1\}$
	\begin{equation}
		\w_i=\sum\limits_{j=1}^n h_{ij}\,\f_j\text{ for }i\in[M].
	\end{equation}
	
	Then watermarks are  added to the host signal to obtain a final copy $\y_i$ for the $i$-th user
	\begin{equation*}
		\y_i=\x+\w_i.
	\end{equation*}
	We assume that $\norm{\w_i}\ll\norm{\x}$, so the added watermark doesn't change the content much.
	
	A coalition of dishonest users $I \subset [M]$ may come together to forge a new copy and redistribute it among other users. They can apply a linear attack, i.e., create a new copy $\y$ as a linear combination of their copies. In addition, they may add a noise vector $\eps$, $\norm{\eps}\ll \norm{\x}$,  to make it harder for the dealer to identify them.
	\begin{equation*}
		\y = \sum\limits_{i\in I}\lambda_i\,\y_i+\eps,
	\end{equation*}
	where $\lambda_i>0$ for each dishonest user in $I$ exactly participates in the attack,
$\lambda_i\in\mathbb{R}$ and $\sum_{i\in I}\lambda_i=1$ to ensure the multimedia content $\x$ not be changed. Especially in
averaging attack, the last condition is $\lambda_i = 1/\abs{I}$ for every $i \in I$ and it  implies that
	\begin{equation*}
		\y = \sum\limits_{i\in I}\lambda_i\,\y_i +\eps= \x+\sum\limits_{i\in I}\lambda_i\,\w_i+\eps.
	\end{equation*}
	Note that
	\begin{equation*}
		\norm{\y-\x}= \norm{\sum\limits_{i\in I}\lambda_i\,\w_i+\eps}\leq \max \norm{\w_i}+\norm{\eps}\ll \norm{\x},
	\end{equation*}
	therefore, $\y$ is close enough to the original signal $\x$.
	
	In order to find the coalition of dishonest users based on the forged copy $\y$, the dealer evaluates 
	\begin{align*}
		s_k&=\langle\y-\x, \f_k\rangle\\&=
		\langle\sum\limits_{i\in I}\lambda_i\,\sum\limits_{j=1}^nh_{ij}\,\f_j+\eps, \f_k\rangle=
		\sum\limits_{i\in I}\lambda_i\,h_{ik}+e_k,
	\end{align*}
	where $e_k=\langle\eps, \f_k\rangle$,
	and forms a syndrome vector $\S=(s_1,\,\ldots,\,s_n)$. The syndrome vector $\S$ can be equivalently defined through the matrix equation
	\begin{equation*}
		\S=H \bLambda^T+\e,
	\end{equation*}
	where $\Lambda=(\lambda_1,\,\ldots,\,\lambda_M)$, $\lambda_i=0$ for $i\notin I$, and $\e=(e_1,\,\ldots,\,e_n)$, $\norm{\e}\leq \norm{\eps}$.
	
	The dealer wants to design a matrix $H$ in such a way, that by observing $\S$ he always can find the support $\supp(\bLambda)$ if the size of the coalition $I$ is at most $t$. The following definition for a noiseless scenario was introduced in~\cite{egorova2019signature}.
	
	\begin{definition}
		A binary $n\times M$ matrix $H$ is called a $t$-multimedia  digital fingerprinting code with complete traceability ($t$-MDF code for short) if  
		for any two distinct coalitions $I$, $I'$, $\abs{I}, \abs{I'}\leq t$, we have
		\begin{equation*}
			H\bLambda^T\neq H\bLambda'^T
		\end{equation*}
		for any real vectors $\bLambda=(\lambda_1,\,\ldots,\,\lambda_M)$ and 
		$\bLambda'=(\lambda_1',\,\ldots,\,\lambda_M') $, such that $\lambda_i\geq 0$, 
		$\lambda_i'\geq 0$, $\sum\limits_{i=1}^M\lambda_i=\sum\limits_{i=1}^M\lambda_i'=1$, 
		$\supp(\bLambda)=I$, $\supp(\bLambda')=I'$.
	\end{definition}
	
	%
	
	Denote the maximal cardinality and the maximal rate of  $t$-MDF code of length $n$ as $M(n, t)$ and $R(n, t)=n^{-1}\log_2M(n,t)$. Denote by $R^*(t)$ and $R_*(t)$ an upper and a lower limits of $R(n, t)$ as $n\to\infty$. It is known that 
	\begin{equation}\label{eq::noiseless bounds}
		\Omega\left(\frac{\log_2t}{t}\right)\leq R_*(t)\leq R^*(t)\leq \frac{\log_2t}{2t}(1+o(1)).
	\end{equation}
	The upper bound of~\eqref{eq::noiseless bounds} can be derived from an upper bound for a binary adder channel from~\cite{d1981coding}. The lower bound is based on the following observation from~\cite{egorova2019signature}. If any $2t$ columns of a binary matrix $H$ are independent over the field of real numbers $\mathbb{R}$, then $H$ is a $t$-MDF code. Since parity check matrices of binary codes with a distance $d>2t$ poses this property, application of Goppa or BCH codes gives an explicit construction with a rate $R_*(t)\geq 1/t$\cite{egorova2019signature}. An improved lower bound $\Omega\left(\frac{\log_2t}{t}\right)$ can be derived from the results of the paper~\cite{bshouty2011parity}, where the authors proved the  existence of binary $n\times M$ matrices, $n^{-1}\log_2M=\Omega(\log_2t/t)$, such that any $2t$ columns are independent over the field $\mathds{Z}_p$, $p>2t$. We note that the latter result was proved with a probabilistic method, i.e. it's not explicit.

	Now we discuss a noisy scenario. In~\cite{fan2020signature} the authors defined $(t, \delta)$-light complete traceability multimedia digital fingerprinting codes and proved that they don't exist. Informally, if some coefficient $\lambda_i$ is sufficiently small, then it is possible to compensate the signal of $i$th user by the noise so that it would be impossible to identify this user. However,  for the case of averaging attacks, when all non-zero coefficients $\lambda_i$ are equal, the situation is different. Let us give the corresponding definition from~\cite{egorova2020existence}.
	\begin{definition}
		A binary $n\times M$ matrix $H$ is called a (Euclidean) $(t, \delta)$-light complete traceability code  if  
		for any two distinct coalitions $I_1$, $I_2$, $\abs{I_1}, \abs{I_2}\leq t$, we have
		\begin{equation*}
			\sigma(H\mid I_1)+\e_1\neq \sigma(H\mid I_2)+\e_2,
		\end{equation*}
		for any real vectors $\e_1, \e_2 \in \mathbb{R}^n$, $\norm{\e_1}, \norm{\e_2}\leq \delta$.
	\end{definition}
	In other words, Euclidean distance between vectors $\sigma(H\mid I_1)$ and $\sigma(H\mid I_2)$, generated by different coalitions $I_1$ and $I_2$, $\abs{I_1}, \abs{I_2}\leq t$, should be big, i.e.
	$$
	\norm{\sigma(H\mid I_1)-\sigma(H\mid I_2)}>2\delta.
	$$
	
	\begin{remark}
		Although an averaging attack is very restrictive for the coalition, in many papers authors consider only them instead of general linear attacks. One of the arguments is that averaging attack is the most fair choice since all the members of a coalition contribute the same proportion of data into a forged copy~\cite{cheng2011anti,trappe2003anti}.  However, in future research it may be reasonable to study a model with different coefficients $\lambda_i$, which are lower bounded by some constant.  
	\end{remark}
	
	Define codes for the case of noise vectors with a bounded cardinality of their support.
	\begin{definition}
		A binary $n\times M$ matrix $H$ is called a Hamming $(t, T)$-light complete traceability code  if  
		for any two distinct coalitions $I_1$, $I_2$, $\abs{I_1}, \abs{I_2}\leq t$, we have
		\begin{equation*}
			\sigma(H\mid I_1)+\e_1\neq \sigma(H\mid I_2)+\e_2,
		\end{equation*}
		for any real vectors $\e_1, \e_2 \in \mathbb{R}^n$, $\abs{\supp(\e_1)}, \abs{\supp(\e_2)}\leq T$.
	\end{definition}
	Equivalently, the number of different coordinates of vectors $\sigma(H\mid I_1)$ and $\sigma(H\mid I_2)$, generated by distinct coalitions $I_1$ and $I_2$, $\abs{I_1}, \abs{I_2}\leq t$, should be big, i.e.
	$$
	\abs{\supp(\sigma(H\mid I_1)-\sigma(H\mid I_2))}>2T.
	$$
	
	Denote the maximal cardinality of Euclidean and Hamming light complete traceability codes of length $n$ by $M_E(n, t, \delta)$ and $M_H(n, t, T)$ respectively. Define the rates of these codes as follows
	$$R_E(n, t, \delta)=\frac{\log_2M_E(n, t, \delta)}{n},$$
	$$R_H(n, t, T)=\frac{\log_2M_H(n, t, T)}{n}.$$
	
	In the following proposition we show an obvious connection between these two families of codes.
	\begin{proposition}\label{pr::connection}
		1. A Hamming $(t, T)$-light complete traceability code $H$ is a Euclidean  $(t, \delta)$-light complete traceability code for $\delta=\sqrt{2T}/(2t(t-1))$. \\
		2. A Euclidean $(t, \delta)$-light complete traceability code $H$ is a Hamming  $(t, T)$-light complete traceability code for $T=\lfloor2\delta^2\rfloor$. \\
		3. The rates of these codes are connected as follows
		\begin{align*}
			&R_E(n, t, \sqrt{2T}/(2t(t-1)))\geq R_H(n, t, T),\\
			&R_H(n, t, \lfloor2\delta^2\rfloor)\geq R_E(n, t, \delta).
		\end{align*}
	\end{proposition}
	\begin{proof}
		1. Assume that a Hamming  $(t, T)$-light complete traceability code $H$ is not a Euclidean  $(t, \sqrt{2T}/(2t(t-1)))$-light complete traceability code, i.e. there exist two coalitions $I_1$ and $I_2$, such that 
		$$
		\norm{\bDelta}\leq 2\delta,\text{ where }\bDelta=\sigma(H\mid I_1)-\sigma(H\mid I_2), \;\delta=\sqrt{2T}/(2t(t-1)).
		$$
		Since the minimal positive value of coordinate $\Delta_i$ is at least $1/(t(t-1))$,  we conclude that there are at most 
		$$4\delta^2t^2(t-1)^2= 2T$$
		coordinates, in which  $\sigma(H\mid I_1)$ and $\sigma(H\mid I_2)$ are different. 
		Hence, there are two vectors $\u_1$, $\u_2$, $\abs{\supp(\u_1)}, \abs{\supp(\u_2)}\leq T$, such that
		$$
		\sigma(H\mid I_1)+\u_1=\sigma(H\mid I_2)+\u_2.
		$$
		Therefore, $H$ is not a Hamming  $(t, T)$-light complete traceability code.
		This contradiction proves the first claim.
		
		2. Assume that a Euclidean $(t, \delta)$-light complete traceability  code $H$ is not a Hamming  $(t, \lfloor2\delta^2\rfloor)$-light complete traceability code, i.e. there exist two coalitions $I_1$ and $I_2$, such that
		$$
		\abs{\supp(\bDelta)} \leq 2T, \text{ where } \bDelta = \sigma(H\mid I_1)-\sigma(H\mid I_2), \;T=\lfloor2\delta^2\rfloor.
		$$
		
		Since the absolute value of every coordinate of the vector $\bDelta$ is at most 1, we have 
		$$
		\norm{\bDelta}\leq \sqrt{2T}\leq 2\delta,
		$$
		which contradicts the definition of Euclidean $(t, \delta)$-light complete traceability codes.
		
		3. Claim 3 is an obvious corollary of claims 1 and 2.
	\end{proof}
	
	%
	%
	%
	In~\cite{egorova2020existence} it was proved that $\liminf_{n\to\infty}R_E(n, t, \delta)\geq \Omega(1/t)$ for constant $\delta$. An upper bound is the same as in the noiseless case, $\limsup_{n\to\infty}R_E(n, t, \delta)\leq \frac{\log_2t}{2t}(1+o(1))$, since the proof works for an averaging attack. Therefore, there is a $\Theta(\log_2t)$ gap between the lower and upper bound. We eliminate this gap in the next section.

	\section{Lower bound on the rate of light complete traceability codes}\label{sec::main result}
	
	In this section we prove
	\begin{theorem}\label{th::main result} For $\tau<1/4$
		\begin{equation}
			\liminf_{n\to\infty}R_H(n, t, \lfloor \tau n\rfloor)\geq \frac{(1-2\tau)\log_2t}{6t}(1+o(1)),\;\;t\to\infty.
		\end{equation}
	\end{theorem}
	
	Combining Theorem~\ref{th::main result} and Proposition~\ref{pr::connection} we obtain the following 
	\begin{corollary}
		For $\delta^2=\alpha n$, $\alpha<1/8t^2(t-1)^2)$, we have
		\begin{equation*}
			\liminf_{n\to\infty}R_E(n, t, \sqrt{\alpha n})\geq \frac{(1-2\tau)\log_2t}{6t}(1+o(1)),\;\;t\to\infty,
		\end{equation*}
		where $\tau=2\alpha t^2(t-1)^2.$
	\end{corollary}

	For the case of small noise $\delta=o(\sqrt{n})$ and $n\to\infty$ a new lower bound has the following form
	$$
	\liminf_{n\to\infty}R_E(n, t, \delta)
    \geq  \lim\limits_{\alpha\to 0}
    \liminf_{n\to\infty}R_E(n, t, \sqrt{\alpha n})
    \geq\frac{\log_2t}{6t}(1+o(1)).
	$$
	It improves the previous lower bound $\Omega(1/t)$ and has the same order $\Theta(\log_2t/t)$ as the upper bound. However, the new bound is not explicit, i.e. there is no effective encoding or decoding algorithm for a new code.
	
	\begin{proof}[Proof of Theorem~\ref{th::main result}]
		Consider a random $n\times M$ matrix $H$, $M=2^{Rn}$, in which every entry is chosen independently and equals 1 with a probability $1/2$. The value of $R$ will be specified later.
		Fix two coalitions $I_1$ and $I_2$, $\abs{I_1}, \abs{I_2}\le t$. Call a row $r$ good, if 
  
$$
 \frac{\sum\limits_{i_1\in I_1} h_{r,i_1}}{\abs{I_1}}\ne \frac{\sum\limits_{i_2\in I_2} h_{r,i_2}}{\abs{I_2}}.
$$ 
Otherwise, we call a row bad. Call a pair of coalitions good, if there are at least  $2T+1$ good rows for them. Otherwise, call such a pair bad. Then the condition that $H$ is a Hamming  $(t, T)$-light complete traceability code is equivalent to the absence of  bad pairs of coalitions.
		
		We say that a bad pair of coalitions $I_1$ and $I_2$  is minimal, if there is no another bad pair of coalitions $I_1'$ and $I_2'$, $I_1'\cup I_2'\subset I_1\cup I_2$. For example, a bad pair of intersecting coalitions $I_1$ and $I_2$ with $\abs{I_1}=\abs{I_2}$ can't be minimal, since it contains another bad pair $I_1\setminus I_2$ and $I_2\setminus I_1$.
		Obviously, to prove that $H$ is a Hamming  $(t, T)$-light complete traceability code it is enough to check that there are no minimal bad pairs of coalitions.
		
		We are going to prove that a mathematical expectation of the number of minimal bad pairs of coalitions is tending to zero as $n\to\infty$. By Markov's inequality this would imply that for big enough $n$ there exists $(t, T)$-light complete traceability Hamming code with the rate $R$ and $T=\lfloor\tau n\rfloor$.
		
		Now we estimate the probability that a row $i$ is bad for coalitions $I_1$ and $I_2$.
		\begin{lemma}\label{lem::bad row}
			The probability that a row is bad for coalitions $I_1$ and $I_2$, $\abs{I_1}=q$, $\abs{I_2}=r$, $q>r$, is upper bounded by $p(q)=q^{-1/3+o(1)}$, $q\to\infty$. For non-intersecting coalitions $I_1$ and $I_2$, $\abs{I_1}=\abs{I_2}=q$, the probability that a row is bad is upper bounded by $p(q)=q^{-1/2+o(1)}$. Moreover,  $p(q)\leq 1/2$ for all $q$.
		\end{lemma}
		\begin{proof}[Proof of Lemma~\ref{lem::bad row}]
			For the case $q=r$, $I_1\cap I_2=\emptyset$,  probability of a bad row is equal to
			\begin{align*}
				2^{-2q}\sum\limits_{i=0}^q\binom{q}{i}\binom{q}{i}=2^{-2q}\binom{2q}{q}=O(q^{-0.5}),
			\end{align*}
			which is 
			not greater than $1/2$ for all $q$.
			
			Now assume that $q>r$.
			Denote the cardinality of the intersection of $I_1$ and $I_2$ as $k$. Consider two cases $q-k>s$ and $q-k\leq s$, $s=q^{2/3}$.
			
			\textit{The first case $q-k>s$}. 
			Note that for any distribution of zeroes and ones in columns from $I_2$ there exists at most one fraction of ones in $I_1\setminus I_2$ which makes the row bad. Hence the probability of obtaining a bad string is upper bounded by
			$$
			\max\limits_{l}\frac{\binom{q-k}{l}}{2^{q-k}}\leq 1/2.
			$$
			For $q\to\infty$ 
			this bound looks as follows
            
			$$
			\max\limits_{l}\frac{\binom{q-k}{l}}{2^{q-k}}<\frac{1+o(1)}{\sqrt{\pi(q-k)/2}}<\frac{1+o(1)}{\sqrt{\pi s/2}}=O(q^{-1/3}),
			$$
            where in the first inequality a Stirling's approximation $\binom{q-k}{(q-k)/2}\sim \frac{2^{q-k}}{\sqrt{\pi(q-k)/2}}$ for a maximal binomial coefficient was used.
			
			\textit{The second case $q-k\leq s$}. 
			Observe that the greatest common divisor $d=(q,r)$ is at most $s$, since $d\leq q - r\leq q - k\leq s$. 
			Since $s_1/q=s_2/r$ implies $(q/d) \mid s_1$ and  $(r/d) \mid s_2$, it is readily seen that for a bad row $i$ the $i$th coordinate in sums $\sum\limits_{j\in I_1}\h_j$ and  $\sum\limits_{j\in I_2}\h_j$  should be divided by $q/d$ and $r/d$ respectively. Therefore,  probability of  a bad row can be upper bounded by the probability $P$ that a binomial random variable $\xi\sim\Bin(q, 1/2)$ is divided by $q/d\geq q^{1/3}$. One can see that $P<1/2$ for $q/d>1$. Now we prove that for $q\to\infty$ the probability $P$ is at most $q^{-1/3+o(1)}$.
			
			By Hoeffding's inequality~\cite{hoeffding1994probability}
			$$
			\Pr(\abs{\xi-q/2}>\sqrt{q\ln q})\leq 2e^{-2\ln q}=O(q^{-2}).
			$$
			Define $S=[\lfloor\frac{q/2-\sqrt{q\ln q}}{q/d}\rfloor, \lfloor\frac{q/2+\sqrt{q\ln q}}{q/d}\rfloor]$.
			Then we can estimate $P$ as follows
			\begin{align*}
				P&= \sum\limits_{l=0}^{d}\Pr(\xi=l\cdot q/d)\\
                &\leq 
				\sum\limits_{l\in S}
				\Pr(\xi=l\cdot q/d) + \Pr(\abs{\xi-q/2}>\sqrt{q\ln q})
				\\
                &\leq \max\limits_x \Pr(\xi=x)\cdot \left\lceil\frac{2\sqrt{q\ln q}+1}{q/d}\right\rceil+O(q^{-2})
                \\
				&\leq \frac{1}{\sqrt{q}}\cdot \left\lceil\frac{2\sqrt{q\ln q}+1}{q/d}\right\rceil+O(q^{-2})
				\\&\leq 
				O(q^{-1/3}\sqrt{\ln q})
				=q^{-1/3+o(1)}.
			\end{align*}
		\end{proof}

		To estimate a mathematical expectation $E$ of  the number of minimal bad pairs of coalitions we iterate over all $< M^{q+r}$ pairs of coalitions having sizes $q$ and $r$, $q>r$, all pairs of non-intersecting coalitions of size $q$,  and over all possible amounts $L<2T+1$ of good rows.
		
		\begin{align*}
			E&<\sum\limits_{0<r\leq q\leq t}M^{q+r}\sum\limits_{L=0}^{2T}\binom{n}{L}(1-p(q))^Lp(q)^{n-L}\\
			&\stackrel{a)}{<}\sum\limits_{q=1}^{t} qM^{2q}(2T+1)\binom{n}{2T}(1-p(q))^{2T}p(q)^{n-2T}\\
			&=\sum\limits_{q=1}^{t} 2^{2qRn}p(q)^n2^{(h(2\tau) + o(1))n}\left(\frac{1-p(q)}{p(q)}\right)^{2\tau n}\\
			&=\sum\limits_{q=1}^{t}2^{A(q)n},
		\end{align*}
		where
		\begin{align*}
			A(q)=2qR +\log_2p(q) + h(2\tau) + 2\tau\log_2\left(\frac{1-p(q)}{p(q)}\right).
		\end{align*}
		In inequality a) we used the fact that $$\binom{n}{L}(1-p(q))^Lp(q)^{n-L}\leq \binom{n}{2T}(1-p(q))^{2T}p(q)^{n-2T},$$ since 
		$2\tau<1/2\leq1-p(q)$ and $2T\leq (1-p(q))n$.
		
		Let $\hat{R}=\min\limits_{q\in[1,t]}-\frac{\log_2p(q) + h(2\tau) + 2\tau\log_2\left(\frac{1-p(q)}{p(q)}\right)}{2q}$. Note that since $2\tau<1-p(q)$ for all $q$ then $\hat{R}>0$. For $R<\hat{R}$ the condition $A(q)<0$ holds, hence, $E\to 0$ as $q\to\infty$, which implies that the rate $\hat{R}$ is achievable. For $t\to\infty$ the minimum would be attained at $q$, which tends to $\infty$, so
		$$
		\hat{R}=\frac{(1-2\tau)\log_2t}{6t}(1+o(1)),\quad t\to\infty.
		$$
		Theorem~\ref{th::main result} is proved.
	\end{proof}

	\section{Conclusion}\label{sec::conclusion}
	In this paper we proved a new lower bound on the rate of Euclidean $(t, \delta)$-light complete traceability codes, which shows that the optimal rate has order $\Theta(\log_2t/t)$. However, the proof uses probabilistic arguments and does not provide an explicit construction with efficient encoding and decoding algorithms. A natural open problem is to design a code with an optimal rate and efficient decoding algorithm.

    Coefficient $\lambda_i$ shows what proportion of the original content was contributed by user $i$ into an illegal copy. It is natural that if the contribution of user $i$ was very small then it will be hard for a dealer to identify such user.  So, another open task is to design a code capable of finding all members of a coalition for an adversarial noise and linear attack, whose coefficients $\lambda_i$ are lower bounded by some constant, i.e. all users, whose contribution was big enough.
	\section*{Acknowledgment}
	The reported study was supported by RFBR and National Science Foundation of Bulgaria (NSFB), project number 20-51-18002, and by RFBR, grant no. 20-01-00559.

		\bibliographystyle{IEEEtran}
	\bibliography{pirates}
	

\end{document}